\documentclass[aps,prd,a4paper,nofootinbib,floatfix,eqsecnum,twocolumn]{revtex4-1} 
\usepackage{amsmath}
\usepackage{graphicx}
\usepackage{color}
\usepackage[normalem]{ulem} 
\usepackage{cancel} 

\usepackage{amsthm} 

\newtheorem*{theorem}{Theorem} 
\newtheorem*{lemma}{Lemma} 
\newtheorem*{conclusion}{Conclusion} 

\allowdisplaybreaks

\newcommand{\be}{\begin{equation}}
\newcommand{\ee}{\end{equation}}

\newcommand{\pa}{\partial}

\newcommand{\bea}{\begin{eqnarray}}
\newcommand{\eea}{\end{eqnarray}}
\newcommand{\ben}{\begin{eqnarray*}}
\newcommand{\een}{\end{eqnarray*}}

\newcommand{\md}{\mathrm{d}} 

\begin{document}

\title{General-relativistic versus Newtonian:
geometric dragging and dynamic anti-dragging in stationary selfgravitating disks in the first post-Newtonian approximation}
 
\author{Piotr Jaranowski}
\affiliation{Wydzia\l~ Fizyki,
Uniwersytet w Bia{\l}ymstoku,
Lipowa 41, 15--424 Bia{\l}ystok, Poland}

\author{Patryk Mach}
\author{Edward Malec}
\author{Micha\l~Pir\'og}
\affiliation{Instytut Fizyki Mariana  Smoluchowskiego, Uniwersytet Jagiello\'nski, {\L}ojasiewicza 11, 30-348 Krak\'{o}w, Poland} 

\begin{abstract}
We evaluate general-relativistic effects in motion of  stationary selfgravitating accretion disks around a Schwarzschild black hole, assuming the  first post-Newtonian (1PN) approximation.
There arises an integrability condition, that leads to the emergence of two types of general-relativistic corrections  to a Newtonian rotation curve.
The  well known geometric dragging of frames accelerates rotation but the hitherto unknown dynamic term, that  reflects the disk structure, deccelerates rotation. The net result can diminish the Newtonian angular velocity  of rotation in a central disk zone but the geometric dragging of frames dominates in the disk boundary zone. Both effects are nonlinear in nature   and they disappear  in the limit of test fluids.  Dust disks can be only geometrically dragged while uniformly rotating gaseous  disk are untouched at the 1PN order. General-relativistic contributions can strongly affect rotation periods in  Keplerian motion  for    compact systems.
\end{abstract}

\maketitle

\section{Introduction}

Rotation curves are important  characteristics of stationary accretion disks. Angular velocities can be determined
from observations of astrophysical systems and they  allow for the direct determination of the central mass when disks are light and their selfgravity can be neglected \cite{Pringle,Novikov-Thorne}. For heavy disks the selfgravity must be included and  even in the Newtonian case one can give only a  rough estimate of the central mass \cite{MMP1}. Additional information on the disk geometry  and modelling would be required to learn more about masses of observed objects. This approach has been  applied to NGC 4258, the unique AGN with a well measured Keplerian rotation curve of the central disk \cite{MMP2}. A well known feature of axially symmetric Newtonian accretion disks is that  rotation curves of fluids  depend only on the distance to the rotation axis and do not depend on the distance to the plane of symmetry. 

The principal aim of this paper is to investigate general-relativistic corrections to Newtonian  rotation curves in systems with stationary accretion disks.  We take into account self-gravity of the accretion disk and   in  numerical analysis  specialize to the Keplerian  rotation law, at the 0PN (Newtonian) level. The 1PN approximation scheme is employed, following Blanchet, Damour, and Sch\"afer \cite{BDS}.

The existing  research on general-relativistic accretion disks focuses mainly  on test systems in a prescribed spacetime geometry. The literature is extensive, to mention a pioneering paper  by Bardeen and Wagoner \cite{Bardeen_Wagoner} and an early review by Novikov and Thorne \cite{Novikov-Thorne}. One of recent reviews is written by Karas, Hure and Semerak \cite{Karas_2004},  Abramowicz and Fragile \cite{Abram}, and  Stergioulas \cite{Sterg}. See also numerous references therein. We should  mention here the work of Fishbone and Moncrief, who studied the influence of the black holes angular momentum onto the disk structure    for  the stationary flow of isentropic fluid in Kerr geometry  \cite{Fishbone-Moncrief}.  

The first  general-relativistic  formulation of thick selfgravitating disks (around a black hole or a star) has been derived by Nishida, Lanza, and Eriguchi \cite{nishida_eriguchi,nishida1}. They found in particular dragging of inertial frames due to rotating toroids, for general-relativistic extensions of two types of Newtonian rotation curves --- uniform angular velocity and constant specific angular momentum.  Ansorg and Petroff \cite{Ansorg}  studied numerically  a disk --- black hole system from a different perspective, focusing on the geometry of the apparent horizon and its parametrization.

The order of this work is following. The relevant equations in 1PN approximation  are given in Sec.\ 2. Section 3 displays final equations, under the simplifying assumption of axial symmetry. In Sec.\ 4 we show  that the consistency of  1PN equations imposes  an integrability condition. That yields  a  dual structure of the  corrections to the Newtonian  rotation curve. One of the  terms  can be recognized as   the well known geometric dragging of frames induced indirectly (via the backreation effect) by the disk rotation. The other depends on the specific  enthalpy and thus it has a dynamic, material character. In Sec.\ 5 we prove that the dynamic term deccelerates rotation, while the geometric effect increases the  angular velocity. A scaling symmetry of Euler equations allows one to find a simple scaling law for the  1PN  angular velocity correction and for its  ratio to the Newtonian angular velocity. The post-Newtonian corrections   are analyzed for fluids and dust, and for different rotation curves.

It is notable that rigid rotation is untouched by 1PN corrections, while dust disks are influenced only by the geometric drag. Section  6 is dedicated to the description of the numerical approach to the problem. The obtained results are discussed in Sec.\ 7.  It appears --- in agreement with analytic results --- that the two effects, geometric and dynamic,  work against each other  and that absolute values of their extrema are   comparable.   As a consequence the net general-relativistic 1PN effect is  weakest (can vanish), paradoxically,  in a central disk zone where the two component parts taken separately   are strongest.  The dynamic component vanishes at the disk boundary; thus the  dragging of frames dominates in the disk boundary zone. The  important feature seen in the 1PN approximation, that may have observational consequences, is that the rotation curve depends on the height above the plane of disk symmetry.  
Finally, we summarize obtained results and point out open questions.

\section{Equations}

Einstein equations, with the signature $(-,+,+,+)$, read
\begin{equation}
R_{\mu \nu} -g_{\mu \nu }{R\over 2} = 8\pi {G\over c^4}T_{\mu \nu },
\label{ee1}
\end{equation}
where $T_{\mu \nu }$ is the stress-momentum tensor.  
The  {\emph{stationary}} metric is given in the form suitable for the 1PN approximation,
in Cartesian coordinates $x=x^1,y=x^2,z=x^3, x^0=ct$, by
\begin{align}
\label{metric}
\md s^2 &= \left( -1-2{ \frac {U \left( x,y,z\right) }{{c}^{2}}}-2{\frac { \left( U\left( x,y,z \right)  \right) ^{2}}{{c}^{4}}}\right) (\md x^0)^2
\nonumber\\[1ex]&\qquad
- 2{\frac {A_i\left(x,y,z\right) }{{c}^{3}}} \md x^i \md x^0
\nonumber\\[1ex]&\qquad
+ \left(1-2{\frac {U \left( x,y,z \right) }{{c}^{2}}}\right)  \left( \md x^2 + \md y^2 + \md z^2\right) .
\end{align}
In the remainder of this Section we use Cartesian coordinates.
 {We employ the stress-momentum tensor of the form}
\be
T^{\alpha\beta} = T^{\alpha\beta}_\textrm{BH} + T^{\alpha\beta}_\textrm{D},
\ee
where $T^{\alpha\beta}_\textrm{BH}$ describes the point particle
(which models the central black hole) at rest located at the origin of the coordinate system
and $T^{\alpha\beta}_\textrm{D}$ is the  {stress}-momentum tensor of the disk matter.
The tensor $T^{\alpha\beta}_\textrm{BH}$  {describing a single point particle}
is proportional to the Dirac delta distribution,
\be
\label{emBH}
T^{\alpha\beta}_\textrm{BH} = \frac{M_{\textrm{c}} c^2}{\sqrt{g}}
\frac{u^\alpha_\textrm{BH}u^\beta_\textrm{BH}}{u^0_\textrm{BH}}
\delta(\mathbf{x}-\mathbf{z}_\mathrm{BH}(t)),
\ee
where $M_{\textrm{c}}$ is the mass parameter of the point particle,
$g:=-\det(g_{\mu\nu})$ and 
$u^\alpha_\textrm{BH}:=\mathrm{d}z^\alpha_\textrm{BH}/(c\,\mathrm{d}\tau_\textrm{BH})$
is the 4-velocity along the particle's world line parametrized by the proper time $\tau_\textrm{BH}$.
 {We} assume that the point particle is located at rest at the origin of the coordinate system,
{therefore} $\mathbf{z}_\mathrm{BH}(t)\equiv\mathbf{0}$; then $T^{\alpha\beta}_\textrm{BH}$ simplifies to
\be
\label{emBHrest}
T^{00}_\textrm{BH} = \frac{M_{\textrm{c}} c^2}{\sqrt{g}} (u^0_\textrm{BH})^2 \delta(\mathbf{x}),
\quad T^{0i}_\textrm{BH} = T^{ij}_\textrm{BH} = 0.
\ee
The disk is made of perfect fluid with a stress-momentum tensor
\be
\label{emD}
T^{\alpha\beta}_\textrm{D} = \rho (c^2+h)u^\alpha u^\beta + p g^{\alpha\beta},
\ee
where $\rho$ is the baryonic rest-mass density, $h$ is the  specific enthalpy,
and $p$ is the  pressure.  
The 4-velocity  {$u^\mu:=\mathrm{d}x^\mu/(c\,\mathrm{d}\tau)$} along the world line of fluid particles
is normalized ($\tau$ is their proper time), $g_{\alpha\beta}u^\alpha u^\beta=-1$ .
 
The proper specific enthalpy $h$ is related with the proper relativistic energy density $e$
through the relation
\be
h = \frac{e+p}{\rho} - c^2.
\ee
We assume the  polytropic equation of state  
\be
e(\rho,S) = \rho c^2 + \frac{K(S)}{\gamma-1}\rho^\gamma,
\ee
where $S$ is the specific entropy of fluid.
Then the following relations hold
\begin{align}
p(\rho ,S) &= \rho \left(\frac{\pa e}{\pa \rho }\right)_{\!S} - e = K(S) \rho^\gamma,
\\[1ex]
h(\rho ,S) &= K(S) \frac{\gamma}{\gamma-1}\rho^{\gamma-1}.
\end{align}
In this paper we assume that the entropy is constant.

The 1PN-accurate  stationary relativistic Euler equation can be  derived  directly from the conservation law,
$\nabla_\alpha T^{\alpha\beta} = 0$, and the continuity of the baryonic current, $\nabla_\alpha\left(\rho u^\alpha\right)=0$.
Alternatively one can employ Eqs.\ (2.18) in \cite{BDS}. The result reads
\begin{widetext}
\begin{multline}
\label{euler1}
\pa_j \left( \rho\,v^iv^j + c^{-2} \rho\,v^j \left(-A_i+v^i(h-6U+\mathbf{v}^2)\right)
+ (1-2c^{-2}U)p\,\delta^j_i \right)
\\[1ex]
= -\left( \rho + M_{\textrm{c}} (1 + c^{-2}U) \delta(\mathbf{x})
+ c^{-2} \left(2p + \rho(h - 2U + 2\mathbf{v}^2)\right) \right)\pa_iU
- c^{-2} \rho\,v^j \pa_i A_j,
\end{multline}
where $v^i:=\mathrm{d}x^i/\mathrm{d}t$ is the coordinate velocity of the fluid particle
and $\mathbf{v}^2:=\delta_{ij}v^iv^j$.
The scalar potential $U$ is the solution of the following 1PN-accurate equation:
\be
\Delta U = 4\pi G \left( \rho + M_{\textrm{c}} (1+c^{-2}U)\delta(\mathbf{x}) + c^{-2}\left(2p + \rho (h-2U+2\mathbf{v}^2)\right) \right),
\ee
\end{widetext}
and the vector potential $A_i$ fulfills the equation
\be
\Delta A_i = -16\pi G \rho v_i,
\ee
where $\Delta $ is the flat laplacian.
Asymptotically we have $|{\bf A}|\propto 4J/R $, where $J$ is the total angular momentum of the configuration
and $R:=\sqrt{x^2+y^2+z^2}$ is the coordinate cylindrical radius. In the case of stationary configurations the 1PN-accurate continuity equation
for the prefect fluid with the stress-momentum tensor \eqref{emD} reads
\be
\label{contEq}
0 = \pa_i(\sqrt{g}u^0 \rho v^i) = \pa_i\Big( \rho v^i + c^{-2} \rho v^i \Big(\frac{1}{2}\mathbf{v}^2-3U\Big) \Big)
+ \mathcal{O}(c^{-4}).
\ee

\section{Axially symmetric disks}

Let us now assume axial and  equatorial symmetry.
We shall replace  the  Cartesian coordinates $(x,y,z)$
by cylindrical ones $(r,\phi,z)$, where $x=r\cos\phi$, $y=r\sin\phi$.
Axial symmetry means that the only non-zero cylindrical component of the 3-vector field ${\bf v }$
and the 3-covector field ${\bf A}$ is the azimuthal component $v^\phi$ and $A_\phi$, respectively:
${\bf A}=A_\phi \md \phi $ and $ {\bf v}=v^\phi \partial_\phi$. These two components and the scalar quantities $\rho$, $p$, $h$, $U$ all do not depend on $\phi$.

 {We split different quantities ($\rho$, $p$, $h$, $U$, and $v^i$) into their Newtonian (denoted by subscript `0')
and 1PN (denoted by subscript `1') parts.
E.g., for the baryonic rest-mass density $\rho$ and the fluid velocity $v^i$ this splitting reads}
\begin{subequations}
\label{density_rotation}
\begin{align}
\rho &= \rho_0 + c^{-2} \rho_1,
\\[1ex]
 \qquad
v^\phi  &= v_0^\phi  + c^{-2} v_1^\phi .
\end{align}
\end{subequations}
Notice that, up to 1PN order,   
\begin{equation}
\label{enthalpy}
\frac{1}{\rho} \partial_i p = \partial_i h_0 + c^{-2} \partial_ih_1
 {+ \mathcal{O}(c^{-4})},
\end{equation}
where the {1PN} correction $h_1$ to the specific enthalpy  
 {can be written as}
\begin{equation}
\label{correction_enthalpy}
h_1= \left( \gamma -1 \right) h_0{\rho_1\over \rho_0}.
\end{equation}
 {One can easily derive from Eq.\ \eqref{enthalpy} useful relations connecting gradients of pressure and specific enthalpy
at the Newtonian and 1PN levels,
\be
\pa_i p_0 = \rho_0 \pa_i h_0, \quad
\pa_i p_1 = \rho_0 \pa_i h_1 + \rho_1 \pa_i h_0.
\ee}

Making use of the introduced above splitting of quantities into Newtonian and 1PN parts
one can extract from Eq.\ \eqref{euler1} the Newtonian- and 1PN-level Euler equations.
The Newtonian equations read
\be
\label{NEuler}
\nabla_j(\rho_0 v_0^i v_0^j) + \pa_i p_0 = -\rho_0 \pa_i U_0 + \pa_i U_0\, M_{\textrm{c}} \delta(\mathbf{x}).
\ee
The 1PN Euler equations take the form
\begin{widetext}
\begin{multline}
\label{1PNEuler}
\nabla_j \left( \rho_0(v_0^i v_1^j + v_1^i v_0^j) + \rho_1 v_0^i v_0^j
+ \rho_0\,v_0^j \left(-A_i+v_0^i(h_0-6U_0+ r^2(v_0^\phi )^2)\right)
+ (p_1 - 2 p_0 U_0)\delta^j_i \right)
\\[1ex]
= -(\pa_iU_1+U_0\pa_iU_0)\,M_{\textrm{c}}\delta(\mathbf{x})
- \left( \rho_1 + 2p_0 + \rho_0 \left(h_0 - 2U_0 + 2 r^2(v_0^\phi )^2\right) \right)\pa_iU_0 - \rho_0 \pa_i U_1
- \rho_0\,v_0^j \pa_i A_j.
\end{multline}
\end{widetext}
The splitting of the potential $U$ into its Newtonian $U_0$ and 1PN $U_1$ parts reads
\be
U = U_0 + c^{-2} U_1.
\ee
The determination of $U_0$ needs only material fluid quantities of zeroth order,
while $U_1$ requires also the 1PN density correction $\rho_1$:
\begin{subequations}
\label{U0U1_bis}
\begin{align}
\label{DeltaU0}
\Delta U_0 &= 4\pi G \left( M_{\textrm{c}} \delta(\mathbf{x}) + \rho_0 \right),
\\
\label{DeltaU1}
\Delta U_1 &= 4\pi G \Big( M_{\textrm{c}} U_0 \delta(\mathbf{x}) 
+ \rho_1 + 2p_0 
\nonumber\\&\qquad
+ \rho_0(h_0-2U_0+2 r^2(v_0^\phi )^2) \Big).
\end{align}
\end{subequations}

The   disk mass at the Newtonian level  is equal to $M_{\textrm{D}}= \int_V \textrm{d}^3x  \rho_0$ and the total mass of the system is $M_{\textrm{c}}+M_{\textrm{D}}$. The 1PN  mass correction $M_\textrm{1PN}$ can be read off from the asymptotic expansion of the correction potential $U_1$. It is given by 
\begin{multline}
M_\text{1PN}= \int_V \md^3x 4\pi G \Big( M_{\textrm{c}} U_0 \delta(\mathbf{x}) + \rho_1 + 2p_0
\\
+ \rho_0(h_0-2U_0+2 r^2(v_0^\phi )^2) \Big).
\end{multline}

The right-hand sides of Eqs.\ \eqref{NEuler}, \eqref{1PNEuler}, and \eqref{U0U1_bis}
contain terms proportional to Dirac delta distribution of the form $f(\mathbf{x})\delta(\mathbf{x})$,
where the function $f$ can be singular at $\mathbf{x}=\mathbf{0}$. We replace these terms by $\mathrm{Pf}_\mathbf{0}(f)\delta(\mathbf{x})$,
where $\mathrm{Pf}_\mathbf{0}$ is the ``Hadamard partie finie'' of the function evaluated at its singular point $\mathbf{x}=\mathbf{0}$.
The operation $\mathrm{Pf}_{\mathbf{x}_0}(f)$ for the function $f$ which is singular at the point $\mathbf{x}_0$ is defined as follows.
Let $\mathbf{n}$ be a unit vector, then one defines $f_\mathbf{n}(\varepsilon):=f(\mathbf{x}_0+\varepsilon\mathbf{n})$.
One expands $f_\mathbf{n}$ into a Laurent series around $\varepsilon=0$:
$$
f_\mathbf{n}(\varepsilon) = \sum_{m=-N}^\infty a_m(\mathbf{n})\varepsilon^n.
$$
The finite part of the function $f$ is defined as the coefficient of $\varepsilon^0$ averaged over all directions:
$$
\mathrm{Pf}_{\mathbf{x}_0}(f) := \frac{1}{4\pi} \oint\md\Omega\,a_0(\mathbf{n}).
$$
This way of regularizing singular functions was commonly used
in numerous derivations of post-Newtonian equations of motion for point-particle systems
(up to the fourth post-Newtonian order \cite{JS2013})
and it is best justified by dimensional regularization \cite{DJS2001,BDEF2004}
(the limit $d\to3$ of the $d$-dimensional version of the $\mathrm{Pf}_{\mathbf{x}_0}$ operation
would give in our computations results identical with those obtained by means of the defined above 3-dimensional version of this operation).

The solution of Eq.\ \eqref{DeltaU0} can be written symbolically in the form
\be
\label{U0}
U_0(\mathbf{x}) = -\frac{GM_{\textrm{c}}}{|\mathbf{x}|} + U_0^\mathrm{D}(\mathbf{x}),
\quad U_0^\mathrm{D}(\mathbf{x}) := 4\pi G (\Delta^{-1}\rho_0)(\mathbf{x}).
\ee
Because $\mathrm{Pf}_\mathbf{0}(1/|\mathbf{x}|)=0$, the term $U_0 \delta(\mathbf{x})$
on the right-hand side of Eq.\ \eqref{DeltaU1} is replaced by $U_0^\mathrm{D}(\mathbf{0})\delta(\mathbf{x})$,
so the regularized form of this equation reads
\begin{multline}
\Delta U_1 = 4\pi G \Big( M_\textrm{c} U_0^\mathrm{D}(\mathbf{0})\delta(\mathbf{x})
 + \rho_1 + 2p_0
\\
+ \rho_0(h_0-2U_0+2 r^2(v_0^\phi )^2) \Big).
\end{multline}
The right-hand side of Eq.\ \eqref{NEuler} contains $\pa_i U_0\delta(\mathbf{x})$,
which is replaced by $\mathrm{Pf}_\mathbf{0}(\pa_i U_0)\delta(\mathbf{x})$.
According to \eqref{U0} one computes
\begin{align*}
\mathrm{Pf}_\mathbf{0}(\pa_i U_0) &= \mathrm{Pf}_\mathbf{0}(GM_\textrm{c} x^i/|\mathbf{x}|^3+\pa_iU_0^\mathrm{D}(\mathbf{x}))
\\[1ex]
&= 0 + \pa_iU_0^\mathrm{D}(\mathbf{0}).
\end{align*}
But   $\pa_iU_0^\mathrm{D}(\mathbf{0})=0$, because of the assumed axial and equatorial symmetry, and the whole term vanishes.
Similarly one can show that all terms with Dirac deltas in Eqs.\ \eqref{1PNEuler} vanish.

The Newtonian Euler equations \eqref{NEuler} in cylindrical coordinates take the form
\begin{subequations}
\label{eulerN}
\begin{align}
\pa_z h_0 &= -\pa_z U_0,
\\[1ex]
\pa_r h_0 - r(v_0^\phi)^2 &=  -\pa_r U_0.
\end{align}
\end{subequations}
The 1PN Euler equations \eqref{1PNEuler} written in cylindrical coordinates read
\begin{widetext}
\begin{subequations}
\label{euler1PN}
\begin{align}
\pa_z h_1 &= -\pa_z U_1 - v^\phi_0 \pa_z A_\phi
- (h_0 - 2 U_0 - 2r^2(v_0^\phi )^2\pa_z U_0
+ 2 U_0 \pa_z h_0,
\\[1ex]\nonumber
\partial_r h_1 &- 2 r v^\phi_0 v^\phi_1
- r (v_0^\phi)^2 (h_0 - 6U_0 +r^2(v_0^\phi )^2
\\
&= -\pa_r U_1 - v^\phi_0 \pa_r A_\phi
- (h_0-2U_0-2r^2(v_0^\phi )^2) \pa_r U_0
+ 2 U_0 \pa_r h_0.
\end{align}
\end{subequations}
One can use Newtonian equations \eqref{eulerN} in order to simplify the 1PN equations \eqref{euler1PN}.
The result is
\begin{subequations}
\label{euler1PN-2}
\begin{align}
-\pa_zh_1 -\pa_z U_1 -v^\phi_0 \pa_z A_\phi - (h_0 + 2 r^2(v_0^\phi )^2) \pa_z U_0 &= 0,
\\[1ex]
-\partial_r h_1 -\pa_r U_1+2r v^\phi_0v^\phi_1 - v^\phi_0 \pa_r A_\phi + r^3 (v_0^\phi)^4
+  r (v^\phi_0)^2 h_0
- 4r (v^\phi_0)^2U_0
- (h_0 + 2 r^2(v_0^\phi )^2)\pa_r U_0  &= 0.
\end{align}
\end{subequations}
It is easy to check, making use of $\partial_z v^\phi_0=0$  {(see the next section)},
that these two equations can be written as 
\begin{subequations}
\label{euler1gradPN}
\begin{align}
\pa_z\Psi &= 0,
\\[1ex]
\partial_r \Psi + 2r v^\phi_0v^\phi_1 + A_\phi \pa_r v^\phi_0
-2r^2\pa_r (v^\phi_0)^2  h_0 &= 0.
\end{align}
\end{subequations}
Here the function $\Psi $ is defined as 
\begin{equation}
\label{Psi}
\Psi = -h_1 - U_1 -v^\phi_0 A_\phi + 2 r^2 (v^\phi_0)^2 h_0 - {3\over 2} h^2_0 - 4 h_0 U_0 - 2 U_0^2
- \int \mathrm{d}r\, r^3(v^\phi_0)^4.
\end{equation}
The only nonzero vectorial component $A_\phi $ satisfies the following equation
\be
\label{Afi}
\Delta A_\phi -2\frac{\pa_rA_\phi }{r}= -16 \pi G r^2 \rho_0 v^\phi_0.
\ee
\end{widetext}

\section{The integrability condition}

Differentiation of  the  Eqs.\ \eqref{eulerN} --- the first one  with respect to $r$ and the second
one with respect to  $z$ --- and subtraction of the obtained equations, lead to $\partial_z v^\phi_0=0$.
This is the consistency relation for the validity of the zeroth order (Newtonian) approximation.
That tells us that the Newtonian part of the rotation curve $v^\phi_0$ is an arbitrary function of the cylindrical radius $r$. It is well known that the requirement of stability imposes additional restrictions, through a growth condition imposed onto specific angular momentum   in Newtonian and relativistic hydrodynamics \cite{Tassoul,Seguin,Friedman_Stergioulas}.

The consistency condition for the 1PN approximation can be obtained from Eq.\ \eqref{euler1gradPN}.
Differentiating the first and second equation with respect $r $ and $z$, respectively, and subtracting the obtained equations, one arrives at 
\begin{equation}
\label{euler1aaPN}
2r v^\phi_0\pa_zv^\phi_1 +  (\pa_r v^\phi_0)(\pa_z A_\phi) 
 -2r^2\pa_r (v^\phi_0)^2  \pa_z h_0 =0.
\end{equation}
This constraint is resolved by 
\begin{equation}
\label{constraint_solution}
v^\phi_1 = -{ A_\phi \over 2r v^\phi_0} \pa_r v^\phi_0
+ \frac{rh_0}{ v^\phi_0} \pa_r(v^\phi_0)^2,
\end{equation}
as can be checked by direct inspection.
We can summarize these results as follows.
\begin{theorem}
The 1PN equations (\ref{euler1gradPN}) reduce to the algebraic equation
$$
\Psi(r,z) = \text{const},
$$
provided that the consistency condition 
\begin{equation}
\label{constraint_solution_th}
v^\phi_1 = -{ A_\phi \over 2r v^\phi_0} \pa_r v^\phi_0
+2 rh_0   \pa_r v^\phi_0
\end{equation}
is satisfied.
\end{theorem}
The interpretation of \eqref{constraint_solution_th}  is straightforward. The first term corresponds to the conventional frame dragging, experienced by isolated test bodies in stationary spacetimes.  The frame dragging  is in this case a backreaction type effect --- rotating disks generate (through Einstein equations) the metric function $A_i$, which in turn influences  the rotation, via the first part of the formula (\ref{constraint_solution_th}).    The other term is purely hydrodynamic and it depends both  on the rotation curve $v^\phi_0$ and on the specific enthalpy in  the Newtonian approximation.   It represents the   direct 1PN reaction   of rotating gas, proportional to the specific enthalpy,  onto its own  rotation.

Notice, that  $A_\phi $ and $v^\phi_1$ change sign when $v^\phi_0\rightarrow - v^\phi_0$,    therefore we consider only the case with  $v^\phi_0 \ge 0$.

\section{Scaling symmetry, Newtonian rotation curves and 1PN approximation}

In what follows we shall explain how  the two terms in the expression of $v^\phi_1$  given in (\ref{constraint_solution_th}) influence rotation.  
It is easy to show that the dynamic part  deccelerates rotation. The specific enthalpy  $h\ge 0$ is nonnegative. Let us suppose a nonincreasing   function $v_0^\phi (r)$, thence the term  {$v^\phi_{1\text{dyn}}:=2rh_0\pa_rv^\phi_0$ is nonpositive --- the   instantaneous 1PN dynamic  reaction slows  the motion: ``anti-draggs'' a system.

We prove that, in contrast to the above, the geometric dragging always increases the speed of rotation. The crucial part of the argument is to show, in the forthcoming lemma,  that the function $A_\phi $ is nonnegative. Taking this for granted and again assuming that $ \pa_r v^\phi_0\le 0$ we infer that $ -{ A_\phi \over 2r v^\phi_0} \pa_r v^\phi_0 \ge 0$.
Thus a moving torus  induces (via backreaction) a geometry, that pushes the torus in the same direction; a rotating  torus bootstrapps itself. 
\begin{lemma}
Assume that $ \rho_0 \ge 0$, $v^\phi_0\ge 0$,
both $ \rho_0$ and $v^\phi_0$ are at least of H\"older class $C^{1,\mu}$,  
and  that $A_\phi$, the solution of \eqref{Afi}, vanishes at infinity like $1/R$.
Then $A_\phi $ is nonnegative. 
\end{lemma}
\begin{proof}
The potential $A_\phi $ vanishes at spatial infinity  and Eq. (\ref{Afi}) with the conditions stated in the Lemma would satisfy assumptions 
of the minimum  principle as stated in \cite{Evans, Trudinger}, save the term proportional to $1/r$, which is singular along the $z$-axis.  Ignoring the latter, one would claim that from the minimum principle $A_\phi \ge 0$. 

Due to the above difficulty, we shall adopt another approach to show that the solution is nonnegative everywhere within the disk.  We shall apply  the method of contradiction. Let the   solution of Eq. (\ref{Afi}) exists on all of ${\bf R}^3$.  Let there exists a  region $\Omega $   that intersects the disk with  $A_\phi \le 0$ and that vanishes at an outer 2-surface $S_\infty $ ($S_\infty $ can be located at spatial infinity.)   The potential $A_\phi $ is at least of class $C^{3,\mu}$, from the embedding theorems \cite{Evans}, and vanishes like $1/R$ at infinity.  The complementary  region (possibly empty)  will be called  $\Omega'$ and $A_\phi >0$ on $\Omega'$. The region  $\Omega $ borders  $\Omega'$ along a boundary $\partial \Omega $   with $A_\phi  =0$.  The surface integrals $\int_{\partial \Omega } \md S^i A_\phi \nabla_iA_\phi $  and $\int_{S_\infty}\md S^i A_\phi \nabla_iA_\phi $  vanish, due to the boundary conditions.  

Multiply Eq.\ (\ref{Afi}) by $A_\phi $ over $\Omega $ and integrate by parts.
This yields 
\be
-\int_{\Omega }  \md V \left( \nabla A_\phi \right)^2  = -16\pi G \int_{\Omega } \md V r^2\rho_0 v^\phi_0A_\phi .
\label{Afi1}
\ee
The two boundary terms, that arise during integration by parts, vanish irrespective of whether $\Omega $ is bounded or  unbounded.
Since $A_\phi $ is differentiable and does not vanish identically, the left hand  side of  (\ref{Afi1})   must be strictly negative. But   if $A_\phi <0$ in $\Omega $, then the right hand side is weakly positive. Thus we get a contradiction; the solution $A_\phi $ cannot be negative within the disk volume.    \end{proof}
Notice that the angular velocity of the fluid in the coordinate frame,
$v^\phi={u^\phi}/{u^0}$,
is equal to  the angular velocity of the  fluid as seen by an observer at rest at infinity. The inverse of $v^\phi $ is proportional to the disk rotation period. Thus the term $v^\phi_1$ is responsible for  the 1PN correction to this period. Since  the dynamic part  $v^\phi_{1\text{dyn}}$ is deccelerating rotation, it increases the rotation period. The  drag term with $   A_\phi $  in turn is positive, which means that it tends to shorten the  rotation period.

One can  find out that Newtonian and 1PN  equations  (\ref{U0U1_bis}),  (\ref{eulerN}),   (\ref{euler1PN})  and  (\ref{Afi}) are invariant under  following scalings:
\begin{widetext}
\begin{align}
\label{scaling_symmetry}
{\bf x'}=\lambda {\bf x}, \quad
M'_{\textrm{c}}=M_{\textrm{c}}, \quad
\rho_0' &= \frac{\rho_0}{\lambda^3}, \quad
h_0'=\frac{h_0}{\lambda }, \quad
v'^\phi_0=\frac{v^\phi_0}{\lambda^{3/2}}, \quad
U_0'=\frac{U_0 }{\lambda},
\nonumber\\
A_\phi'=\frac{A_\phi }{\sqrt{\lambda}}, \quad
\rho_1' &= \frac{\rho_1}{\lambda^4}, \quad
h_1'=\frac{h_1}{\lambda^2 }, \quad
v'^\phi_1=\frac{v^\phi_1}{\lambda^{5/2} }, \quad
U_1'=\frac{U_1 }{\lambda^2} .
\end{align}
\end{widetext}
These scalings imply that  under this scaling operation:
\begin{enumerate}
\item the Newtonian disk mass as well as the total Newtonian  mass  are invariant
\item the Newtonian pressure scales as   $p_0'={p_0}/{\lambda^4}$
\item the 1PN mass correction scales according to $  M_\text{1PN}'={  M_\text{1PN}}/{\lambda}$
\end{enumerate}
The scaling of specific enthalpy entails the change of the equation of state.
The coefficient $K$ in the polytropic equation of state $p=K\rho^{\gamma }$ has to scale according to $K'=K\lambda^{3\gamma -4}$.
If we assume the Keplerian rotation law (see Sec.\ VII), then the parameter $\omega_0'=\omega_0$.
The physical sense of this scaling is that the change of distances while keeping masses,
associated with the appropriate adaptation of the equation of state and the rotation law, yields a new solution of the equations.

In our context, since  the Newtonian Euler equations represent the zeroth order approximation to the general-relativistic theory,
we have to demand that the inner boundary of the disk \  is located outside of the Schwarzschild horizon,
 i.e., $(1-U/c^2)r_\text{in}>{2GM_{\textrm{c}}}/{c^2}$ (or even  $(1-U/c^2)r_\text{in}>{6GM_{\textrm{c}}}/{c^2}$),
while the speed of sound and the linear velocity $rv^\phi_0$ do not exceed the speed of light.
These conditions limit our freedom in specifying the scaling parameter $\lambda $.

Equation (\ref{constraint_solution_th}) serves as the consistency condition;
one easily finds that its right-hand side scales like $1/\lambda^{5/2}$,
in agreement with the scaling of $v^\phi_1$ in (\ref{scaling_symmetry}).  
Collecting all these facts together, one immediately proves following
\begin{conclusion}
Let a solution of (\ref{eulerN}), 1PN equations (\ref{euler1PN}) and Eq.\ (\ref{Afi})
be obtained through rescalings  specified in (\ref{scaling_symmetry}).
Then  its first post-Newtonian correction to the angular velocity 
scales as follows:
$$
\frac{v'^{\phi}_1}{c^2v'^{\phi}_0} = \frac{1}{\lambda}\frac{v^\phi_1}{c^2v^\phi_0}.
$$
\end{conclusion}

Dust disks and uniformly rotating gas toroids are distinguished in a sense that becomes evident from the following description.   

i) {\it Uniform rotation.} The correction term $v_1^\phi$ is proportional to the derivative of the angular velocity.
Since  for  uniformly  rotating systems $v^\phi_0= \mathrm{const}$, $v^\phi_1$ strictly  vanishes.
We conclude that in particular the rotation periods of rigidly rotating disks do not change in the 1PN perturbation order.
That adds to the exceptional status of uniformly rotating  disks, which are already known to minimize the total mass-energy
for a given baryon number and total angular momentum \cite{Boyer_Lindquist,Sterg}.
 
ii) {\it Dust.} The pressure and specific enthalpy $h_0$ do vanish in the case of dust,
which means that dust disks are  exposed only to the frame dragging. The behaviour of dust and gas disks are clearly different.

The rotation curve $v^\phi _0$ in the Newtonian approximation depends only on  the distance from the rotation axis. That means that even thick and selfgravitating Newtonian disks rotate uniformly along circles $z=\mathrm{const}$ on  cylinders of constant $r$. In contrast to that, the first post-Newtonian  correction $v_1^\phi $ to the rotation curve is strictly determined by other quantities and depends both on $r $ and $z$. The effective angular velocity is generically not constant along circles of constant height on  cylinders with fixed values of the coordinate  radius  $r$ and  of   the circumferential  radius  $(1-U/c^2)r$.

\section{On  numerical method}

In the following we will work in cylindrical $(r, \phi, z)$ or spherical $(R, \theta, \phi)$ coordinates. For convenience, we will denote $\mu = \cos \theta$. We will also abuse the notation by reserving the same symbol for the given quantity in both coordinate systems.

We employ an iterative, Self-Consistent Field (SCF) type method based on solving two elliptic equations in axial symmetry: the standard scalar Poisson equation
\begin{equation}
\label{zya}
\Delta \Phi = f(r, z),
\end{equation}
and the vector equation of the form
\begin{equation}
\label{zyb}
\Delta \mathbf{A} = f(r, z) \mathbf{v},
\end{equation}
where in cylindrical (or spherical) coordinates the only nonvanishing components of  fields $\mathbf{A}$ and $\mathbf{v}$ are $A_\phi \, \md \phi$ and $v_\phi \, \md \phi$, respectively. It is also assumed that $\partial_\phi A_\phi = \partial_\phi v_\phi = 0$. It is an easy exercise to check that $A_\phi$ satisfies
\begin{multline}
\label{aaa}
\Delta \left( \frac{\cos \phi \, A_\phi}{r} \right) = \frac{f(r, z) \cos \phi \, v_\phi}{r}
\\
\mathrm{or} \quad \Delta \left( \frac{\sin \phi \, A_\phi}{r} \right) = \frac{f(r, z) \sin \phi \, v_\phi}{r}.
\end{multline}
This leads to the equation of the form
\[ \Delta A_\phi - \frac{2 \partial_r A_\phi}{r} = f(r, z) v_\phi.  \]

These equations are solved by expanding the appropriate Green functions in Legendre functions. Although many drawbacks of such an approach are known, we prefer to follow it because of its conceptual simplicity. The following equations are known in the literature (in this or a similar form); we prefer to discuss them here for completeness.

The Green function of the flat 3-dimensional laplacian corresponding to a solution that vanishes asymptotically has the standard expansion
\begin{widetext}
\begin{equation}
\label{aae}
- \frac{1}{4 \pi |\mathbf x - \mathbf x^\prime|} = - \frac{1}{4 \pi} \sum_{j=0}^\infty \frac{R_<^j}{R_>^{j+1}}   \left\{ P_j(\mu) P_j(\mu^\prime) + 2 \sum_{m = 1}^j \frac{(j - m)!}{(j + m)!} P_j^m(\mu) P_j^m(\mu^\prime) \cos [ m (\phi - \phi^\prime)] \right\},
\end{equation}
where $R_{>(<)}$ denotes the larger (smaller) of the two radii $R$ and $R^\prime$. Using Eq.~(\ref{aae}) one can write the solution of Eq.~(\ref{zya}) as
\end{widetext}
\begin{equation}
\label{zzzr}
 \Phi(R,\mu) = - \frac{1}{2} \sum_{j=0}^\infty P_j(\mu) \left[ \frac{1}{R^{j+1}} E_j(R) + R^j F_j(R) \right],
\end{equation}
where
\begin{equation}
\label{zzzs}
E_j(R) = \int_0^R \md R^\prime {R^\prime}^{j+2} \int_{-1}^1 \md \mu^\prime P_j(\mu^\prime) f(R^\prime, \mu^\prime)
\end{equation}
and
\begin{equation}
\label{zzzt}
F_j(R) = \int_R^\infty \md R^\prime \frac{1}{{R^\prime}^{j-1}} \int_{-1}^1 \md \mu^\prime P_j(\mu^\prime) f(R^\prime, \mu^\prime).
\end{equation}
Note that if $f(R,\mu)$ is equatorially symmetric (i.e., it is an even function of $\mu$), integrals with $P_{2j + 1}(\mu)$ vanish. In this case it is also enough to integrate with respect to $\mu$ over $0 \le \mu \le 1$. The numerical implementation of the above formulas is straightforward; it is discussed for instance in~\cite{steinmetz}.

Equation (\ref{zyb}) can be solved in a similar fashion, which is equivalent to finding of a suitable expansion of the Green function for the operator $\Delta - (2/r)\partial_r$. One can start with Eq.~(\ref{aaa}). Using Eq.~(\ref{aae}), the solution for $\cos \phi \, A_\phi/(R \sqrt{1 - \mu^2})$ can be written as 
\begin{widetext}
\begin{eqnarray*}
\frac{\cos \phi \, A_\phi (R, \mu)}{R \sqrt{1 - \mu^2}} & = & - \frac{1}{4 \pi} \int_0^\infty \md R^\prime R^\prime \int_{-1}^1 \frac{\md \mu^\prime}{\sqrt{1 - {\mu^\prime}^2}} \int_0^{2 \pi} \md \phi^\prime \cos \, \phi^\prime f(R^\prime, \mu^\prime) v_\phi (R^\prime, \mu^\prime) \\
& & \times \sum_{j=0}^\infty \frac{R_<^j}{R_>^{j+1}}   \left\{ P_j(\mu) P_j(\mu^\prime) + 2 \sum_{m = 1}^j \frac{(j - m)!}{(j + m)!} P_j^m(\mu) P_j^m(\mu^\prime) \cos [ m (\phi - \phi^\prime)] \right\}.
\end{eqnarray*}
The integral
\[ I = \int_0^{2 \pi} \md \phi^\prime \cos \phi^\prime \cos [ m (\phi - \phi^\prime) ] \]
can be easily evaluated. For $m = 1$ one has $I = \pi \cos \phi$. For $m \neq 1$ one has
\[ I = \frac{2 m \sin (m \pi) \cos [m (\pi - \phi)] }{(m - 1)(m + 1)}. \]
This yields
\begin{eqnarray*}
\frac{\cos \phi \, A_\phi (R, \mu)}{R \sqrt{1 - \mu^2}} & = & - \frac{1}{2} \cos \phi \int_0^\infty \md R^\prime R^\prime \int_{-1}^1 \frac{\md \mu^\prime}{\sqrt{1 - {\mu^\prime}^2}} f(R^\prime, \mu^\prime) v_\phi(R^\prime, \mu^\prime) \\
& & \times \sum_{j=1}^\infty \frac{R_<^j}{R_>^{j+1}} \frac{1}{j(j+1)} P_j^1(\mu) P_j^1 (\mu^\prime).
\end{eqnarray*}
And finally
\begin{equation}
\label{aab}
A_\phi (R, \mu) = -\frac{1}{2} \sqrt{1 - \mu^2} \sum_{j=1}^\infty \frac{1}{j(j+1)} P_j^1(\mu) \left[ \frac{1}{R^j} C_j(R) + R^{j+1} D_j(R) \right],
\end{equation}
where
\begin{equation}
\label{aac}
C_j(R) = \int_0^R \md R^\prime {R^\prime}^{j+1} \int_{-1}^1 \frac{\md \mu^\prime}{\sqrt{1 - {\mu^\prime}^2}} P_j^1(\mu^\prime) f(R^\prime, \mu^\prime) v_\phi (R^\prime, \mu^\prime)
\end{equation}
and
\begin{equation}
\label{aad}
D_j(R) = \int_R^\infty \md R^\prime \frac{1}{{R^\prime}^j} \int_{-1}^1 \frac{\md \mu^\prime}{\sqrt{1 - {\mu^\prime}^2}} P_j^1(\mu^\prime) f(R^\prime, \mu^\prime) v_\phi (R^\prime, \mu^\prime).
\end{equation}
\end{widetext}
Of course $d\mu/\sqrt{1 - \mu^2} = - \md \theta$. Note that
\[ P_{2j}^1 (\mu) = \sqrt{1 - \mu^2} \frac{\md}{\md \mu} P_{2j}(\mu) \]
(where we use the convention without Condon--Shortley's phase).  The derivatives $\md P_{2j}(\mu)/\md \mu$ are odd functions of $\mu$ and  $f (R,\mu) v_\phi (R,\mu)$ are symmetric under reflection $z\rightarrow -z$. Therefore all integrals with $P^1_{2j}$ vanish.

Note also that Eq.~(\ref{aaa}) is equivalent to the equation
\[ \Delta f + \frac{2}{R}\partial_R f - \frac{2}{R^2} \mu \partial_\mu f  = S(R,\mu), \]
that appears in \cite{bonazzola, komatsu, nishida1,  {nishida_eriguchi}}. In cylindrical coordinates the above equation can be rewritten as
\[ \partial_r^2 f + \partial_z^2 f + \frac{3}{r} \partial_r f = S(r, z). \]
Setting $\psi = r \cos \phi \, f$, it is easy to show that $\psi$ satisfies
\[ \Delta_{(r, z, \phi)} \psi = r \cos \phi \left( \partial_r^2 f + \partial_z^2 f + \frac{3}{r} \partial_r f \right) = r \cos \phi \, S(r, z).  \]

In the description of the numerical method given below we specialize to the Keplerian rotation law $v^\phi_0 = \omega_0/r^{3/2}$ and polytropic equations of state $p_0 = K \rho_0^{5/3}$.

The method of obtaining the solutions of the 0-th order approximation (Newtonian solutions) was described in detail in \cite{ MMP1}. Equations to be solved are
\begin{equation}
\label{zzzw}
\Delta U^{\textrm{D}}_0 = 4 \pi G \rho_0
\end{equation}
and
\begin{equation}
\label{zzze}
h_0 + \Phi_{\textrm{c}} - \frac{G M_\textrm{c}}{R} + U^{\textrm{D}}_0  = C,
\end{equation}
where the centrifugal potential is given by
\begin{equation}
\label{zzza}
 \Phi_{\textrm{c}} = - \int^r \md r^\prime r^\prime (v^\phi_0(r^\prime))^2 = \frac{\omega_0^2}{r},
\end{equation}
and the specific enthalpy reads $h_0 = 5K \rho_0^{2/3}/2$. These equations are solved iteratively: in each iteration step one obtains a solution for $U^{\textrm{D}}_0$, basing on the previous density distribution. A new distribution of the enthalpy $h_0$ (or, equivalently, $\rho_0$) is then computed from Eq.~(\ref{zzza}). Also, in each iteration step we renormalize the constants $C$ and $K$, so that the resulting disk has the prescribed values of the inner and outer radii and the maximum density. We use a spherical numerical grid. The solution of Eq.~(\ref{zzzw}) is computed by truncating expansion (\ref{zzzr}) at a sufficiently large number of Legendre polynomials. Integrals (\ref{zzzs}) and (\ref{zzzt}) are computed using standard quadrature formulas.

The 1PN corrections are obtained by solving equations
\begin{widetext}
\begin{equation}
\label{zzzb}
\Delta A_\phi - \frac{2 \partial_r A_\phi}{r} = -16 \pi G \omega_0 \sqrt{r} \rho_0,
\end{equation}
\begin{equation}
\label{zzzc}
 \Delta U_1 = 4 \pi G \left( M_{\textrm{c}} U^{\textrm{D}}_0(\mathbf{0}) \delta(\mathbf x) + \rho_1 + 2 K \rho_0^{5/3} + \rho_0 \left( h_0 - 2 U_0 + 2 \omega_0^2 r^{-1} \right) \right),
\end{equation}
\begin{equation}
\label{zzzd}
h_1 = - U_1 - A_\phi \omega_0 r^{-3/2} + 2 h_0 \omega_0^2 r^{-1} + \frac{1}{2} \omega_0^4 r^{-2} - \frac{3}{2} h_0 - 4 h_0 U_0 - 2 U_0^2 - C_1
\end{equation}
\end{widetext}
for $A_\phi$, $U_1$, $h_1$ and $\rho_1$. Here the Newtonian gravitational potential is given by $U_0 = -G M_{\textrm{c}}/R + U^{\textrm{D}}_0$. The 1PN corrections to the density $\rho_1$ and the enthalpy $h_1$ are related by $h_1 = 5K \gamma \rho_0^{2/3} \rho_1/3$.

Once the 0-th order approximation is known, the potential $A_\phi$ can be obtained from Eq.~(\ref{zzzb}) using the expansion given by Eq.~(\ref{aab}). In the next step we iterate Eqs.~(\ref{zzzc}) and (\ref{zzzd}) in a similar way to that used to obtain the 0-th order solution.

The term $4 \pi G M_{\textrm{c}} U_0 (\mathbf{0}) \delta(\mathbf x)$ on the right-hand side of Eq.~(\ref{zzzd}) yields the term $- G M_{\textrm{c}} U^{\textrm{D}}_0 (\mathbf{0}) / R$ in the solution for $U_1$. Note that, although it is convenient to exclude the origin $\mathbf x = \mathbf{0}$ from the numerical grid, the value $U^{\textrm{D}}_0(\mathbf{0})$ can be still easily computed as
\[ U^{\textrm{D}}_0 (\mathbf{0}) = - G \int \md^3 x \frac{\rho_0}{R}. \] 

In each of the iterations of the Newtonian scheme the value of the enthalpy is obtained from Eq.~(\ref{zzze}). We set $\rho_0 = h_0 = 0$ whenever this equation yields a negative value of $h_0$. This is the key element of our implementation of the free-boundary SCF-type scheme, that allows us to compute the true boundary of the disk. Similarly in the postnewtonian scheme we obtain the value of $h_1$ form Eq.~(\ref{zzzd}). We set $h_1 = \rho_1 = 0$, whenever $h_0 + h_1/c^2 \le 0$ or $\rho_0 + \rho_1/c^2 \le 0$.

The constant $C_1$ appearing in Eq.~(\ref{zzzd}) is renormalized so that in each iteration the correction $h_1$ to the enthalpy   vanishes at the outer end of the  Newtonian disk, in the plane $z=0$. In this way a full post-Newtonian solution is obtained for a specified value of the outer radius.

\section{Discussion of numerical results}

We report  in this Section  numerical results on modelling black hole-disk systems. Disks's masses  are taken in the range  1--2$\times M_{\textrm{c}}$,  where $M_{\textrm{c}}$ is the mass of the central black hole. The Keplerian rotation law   and the polytropic equation of state are assumed, as in the preceding Section.

We carefully choose parameters in all forthcoming examples so that  the 1PN approximation can be valid.
The largest linear velocities $(1-U)rv^\phi $ at the inner part of the disks are of the order of one tenth of the speed of light.
The graphs of the normalized Newtonian and 1PN potentials, ${U_0}/{c^2}$ and ${U_1}/{c^4}$ respectively, are displayed in Figs.\ 1 and 2.
It is clear from Fig.\ 1 that $1\gg{|U_0|}/{c^2}$.
The comparison of Figs.\ 1 and 2 implies that the 1PN correction ${|U_1|}/{c^4}$
constitutes about one hundredth of the main Newtonian contribution.  

\begin{figure}
\begin{center}
\includegraphics[width=0.5\textwidth]{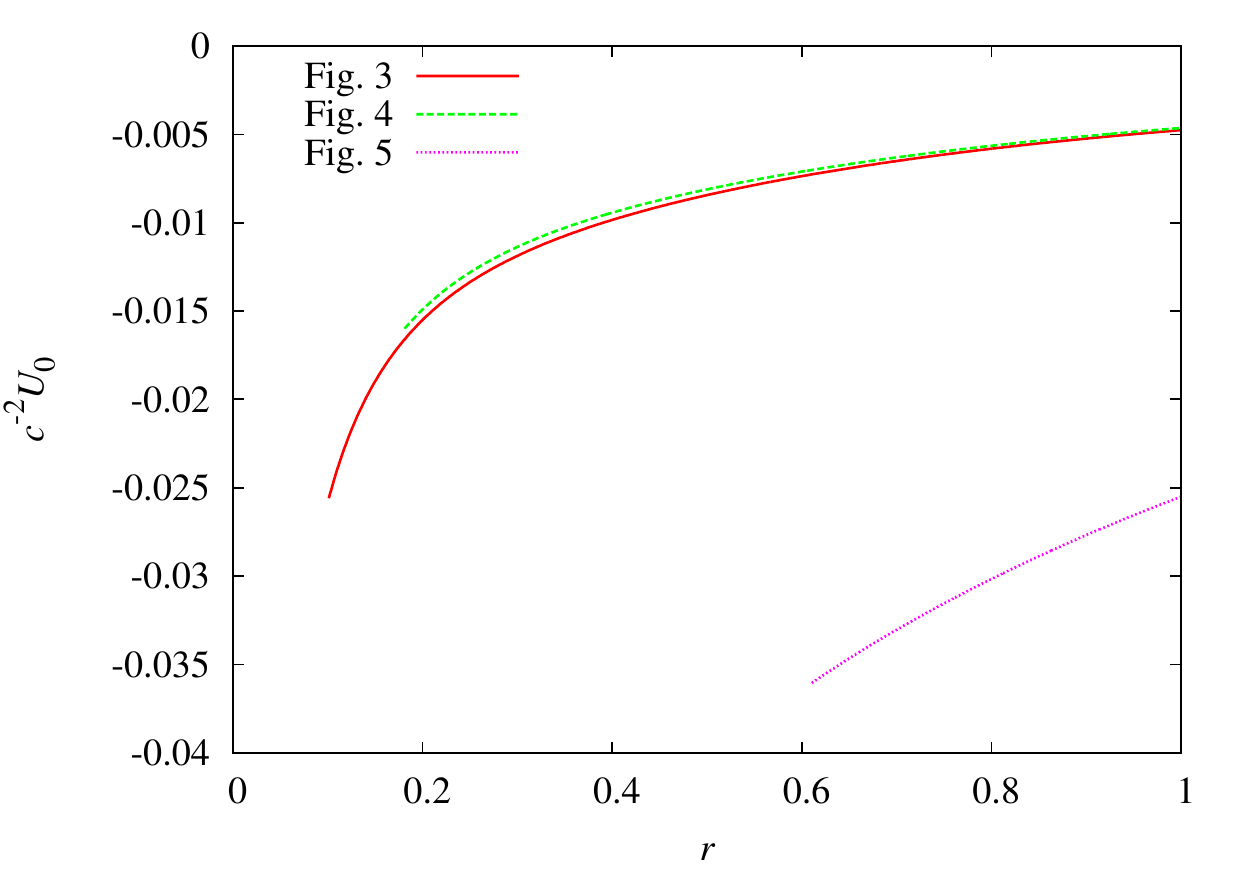}
\end{center}
\caption{\label{fig:1}
The ordinate shows the value of ${U_0}/{c^2}$ along the ray $\phi = \mathrm{const}$ on the central plane $z=0$
and the abscissa the normalized coordinate distance $r/r_\mathrm{out}$ from the center.
The red line corresponds to the solution described in Fig.\ 3,
while the green and violet lines refer to solutions described in Figs.\ 4 and 5, respectively.}
\end{figure}

\begin{figure}
\begin{center}
\includegraphics[width=0.5\textwidth]{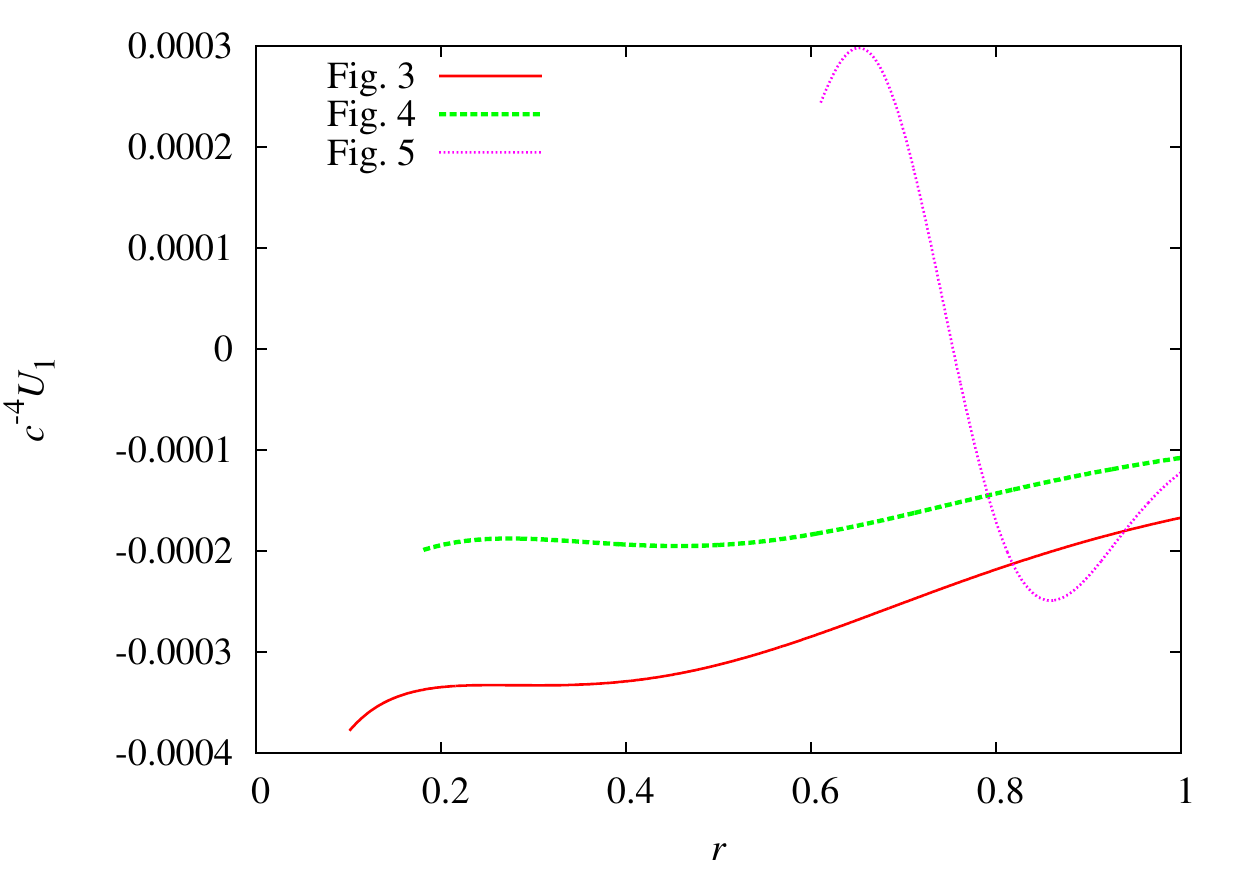}
\end{center}
\caption{\label{fig:2}
The ordinate shows the value of  ${U_1}/{c^2}$ along the ray $\phi = \mathrm{const}$ on the central plane $z=0$
and the abscissa the normalized coordinate distance $r/r_\mathrm{out}$ from the center.
The red line corresponds to the solution described in Fig.\ 3,
while the green and violet lines refer to solutions described in Figs.\ 4 and 5, respectively.}
\end{figure}

The color shaded palettes in Figs.\ 3--5 describe the ratio $v_1^\phi/(c^2v^\phi_0)$ within the disk volume.
In all examples $R_{\textrm{S}}:=2GM_{\textrm{c}}/c^2$ is the Schwarzschild radius of the central black hole.

Keplerian disks  are influenced by both the geometric dragging and the anti-dragging dynamic effects.
 In Fig.\ 3 one sees an inner zone shifted to the centre
with the prevailing braking component (the Newtonian  angular velocity  $v_0^\phi$
is diminished by  the 1PN correction up to 0.3\%),
and the outer part where the drag component dominates
(the  Newtonian} angular velocity  $v_0^\phi$ is enhanced by up to 0.4\%).
The innermost part of the disk is at the coordinate distance $r_\text{in}=25 R_{\textrm{S}}$ 
from the central black hole, while the outermost disk boundary is at $r_\text{out}=250 R_{\textrm{S}}$.
Notice that the circumferential  radius $(1-{U}/c^2)r+\mathcal{O}(c^{-4})$ is well approximated by the coordinate $r$ due to the smallness of potentials.

Our numerical investigation suggests that one can find an infinite number of similar configurations simply
by moving out the inner disk positions up to  $18\%$ of the coordinate size.
Somewhere between the rescaled inner boundary position 0.18 and 0.19     the character of the picture changes 
--- the anti-dragging nowhere dominates and all parts of disks are dragged forward.
This limiting-type configuration is shown in Fig.\ 4.
It remains to be explained in what circumstances the dynamic effect can overcome the geometric dragging.
Our empirical observation is that if the relative width  $w:=(r_\text{out}-r_\text{in})/r_\text{out}>0.2$,
then the latter effect is stronger, and the smaller $w$ the smaller dynamic braking.

We found in a number of examples that the maximal value  of the ratio of the 1PN corrections
to the  Newtonian angular velocities ${v^\phi_1}/({c^2v^\phi_0})$ can achieve a few percents.  
Figure 5 describes a strongly relativistic  compact system,
with the disk width smaller than  $19.5R_{\textrm{S}}$ and relatively large 1PN effects.
The dynamic anti-dragging effect manifests here only by a slight diminishing of the ratio ${v^\phi_1}/({c^2v^\phi_0})$
within the inner part of the central zone.

We already pointed out in Sec.\ V, that one can rescale a given solution within the 1PN approximation
according to the recipe defined in \eqref{scaling_symmetry}.
Physical distances do change under rescalings, but the central mass and the disk mass are invariant.
At the same time the relative velocity correction scales as ${v'^{\phi}_1}/{v'^{\phi}_0 }\propto{1}/{\lambda }$;
thus in principle one can generate from a given solution a sequence of configurations with identical masses,
different geometrical distances and with a different ratio $v^{\phi}_1/(c^2v_0^\phi)$.
One can do that with any of the three already depictured disk configurations.
That suggests, in particular,  that one can generate disk systems with a very large 1PN correction to the angular velocity.
There is, however, a question whether these rescaled solutions can be regarded as being tangent
to the solution of the exact general-relativistic hydrodynamics.
The answer to that cannot be found in the 1PN analysis, but must be sought with the exact general-relativistic treatment.
We leave this problem for future investigation.

\begin{figure}
\begin{center}
\includegraphics[width=0.5\textwidth]{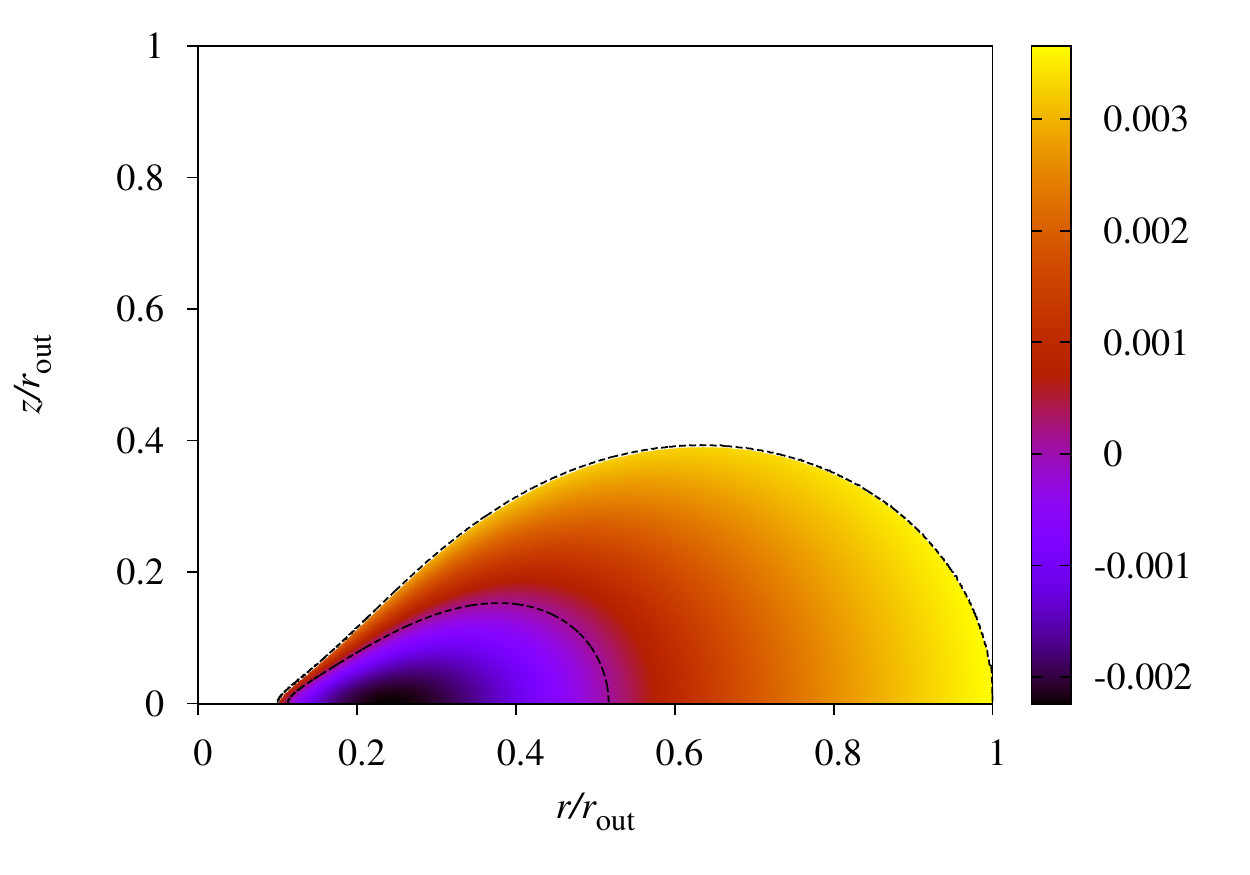}
\end{center}
\caption{\label{fig:3}
 The ratio $v_1^\phi/(c^2v^\phi_0)$ within the disk volume.
The ordinate shows the height of disk and the abscissa the coordinate distance from the center,
in the rescaled unit system $r/r_\mathrm{out}$, where $r_\text{out}\equiv 250 R_{\textrm{S}}$.
The inner disk boundary is located at  $r_\text{in}=25R_{\textrm{S}}$,
the outer boundary $r_\text{out}=250R_{\textrm{S}}$.
The Newtonian disk mass $M_{ \textrm{D} }=1.47\times M_{\textrm{c}}$.
The black region is dominated by the dynamic braking;
its boundary is denoted by the short broken line, where the 1PN velocity correction $v^\phi_1$ vanishes.
The drag is strongest at the disk boundary.
Here $\omega_0=0.85\sqrt{2GM_{\textrm{c}}}$.}
\end{figure}

\begin{figure}
\begin{center}
\includegraphics[width=0.5\textwidth]{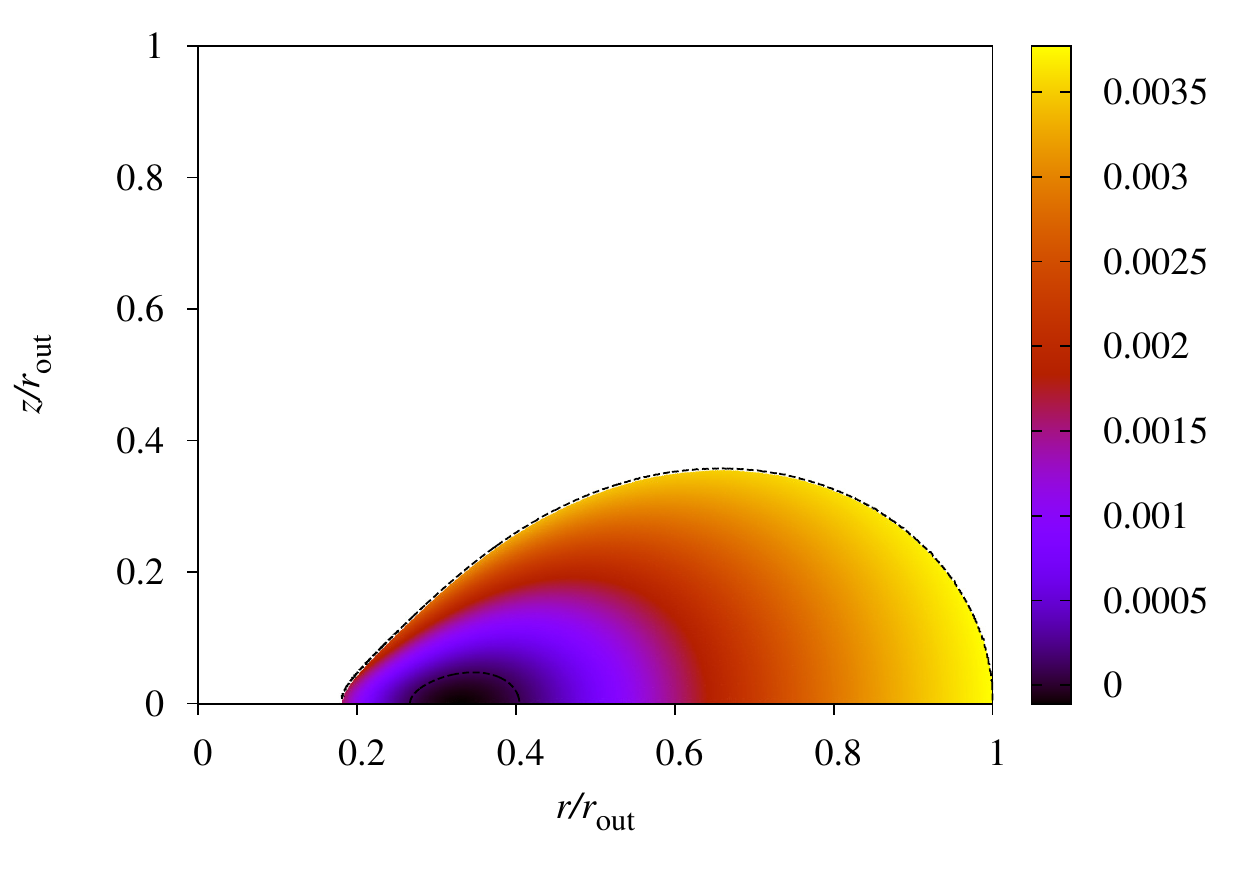}
\end{center}
\caption{\label{fig:4}
 The ratio $v_1^\phi/(c^2v^\phi_0)$ within the disk volume for the limiting-type system.
The ordinate shows the height of disk and the abscissa the coordinate distance from the center,
in  the rescaled unit system $r/r_\mathrm{out}$, where  $r_\mathrm{out}\equiv 450 R_{\textrm{S}}$.
The inner disk boundary is located at  $r_\text{in}=81R_{\textrm{S}}$,
the outer boundary $r_\text{out}=450 R_{\textrm{S}}$.
The Newtonian disk mass $M_{ \text{D} }=1.4 \times M_{\textrm{c}}$.
The black region is dominated by the dynamic braking; its boundary is denoted by the short broken line,
where the 1PN velocity correction $v^\phi_1$ vanishes. The drag is strongest at the disk boundary.
Here $\omega_0=1.188\sqrt{2GM_{\textrm{c}}}$.}
\end{figure}

\begin{figure}
\begin{center}
\includegraphics[width=0.5\textwidth]{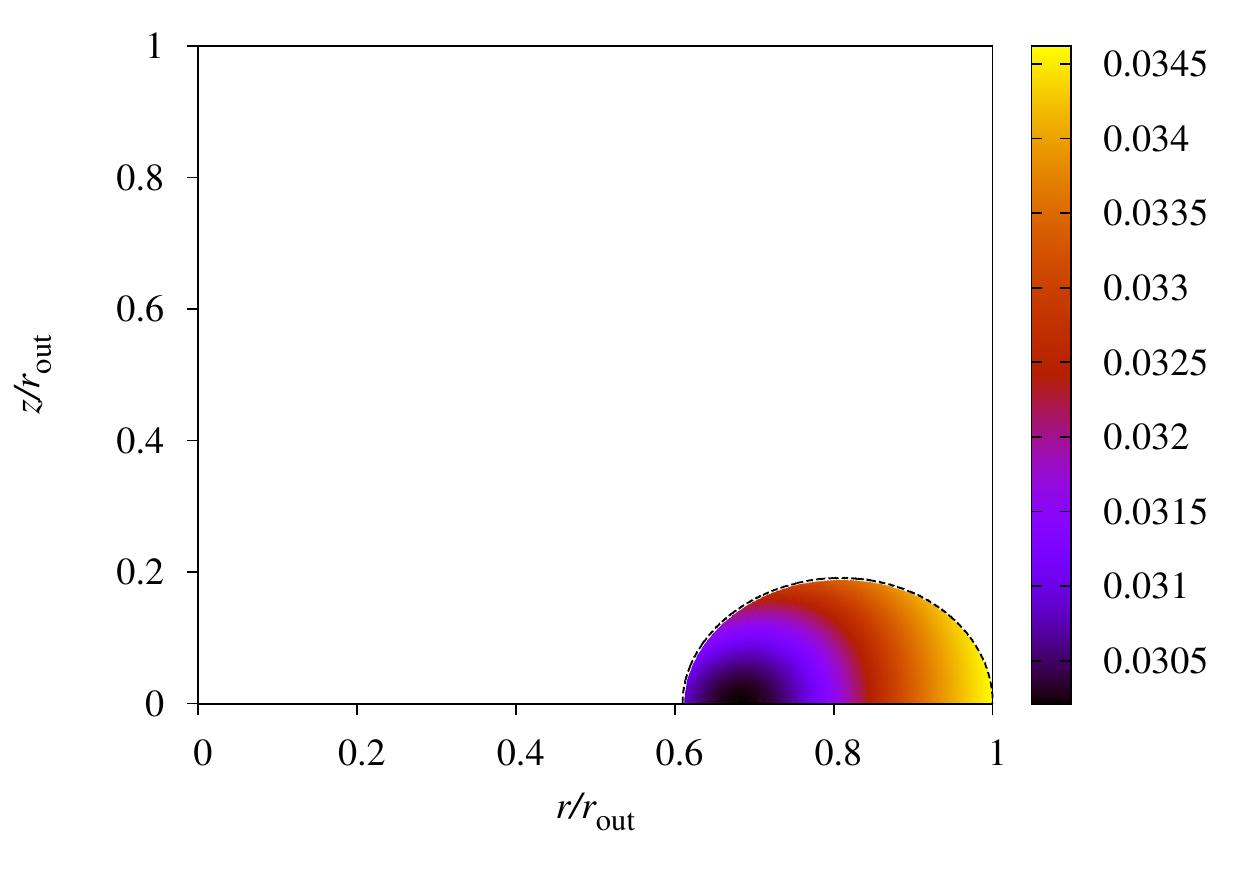}
\end{center}
\caption{\label{fig:5}
 The ratio $v_1^\phi/(c^2v^\phi_0)$ within the disk volume.
The ordinate shows the height of disk and the abscissa the coordinate distance from the center,
in the rescaled unit system $r/r_\mathrm{out}$, where  $r_\mathrm{out} = 50 R_{\textrm{S}}$.
The inner boundary is located at  $r_\text{in}=30.5R_{\textrm{S}}$, the outer boundary $r_\text{out}=50 R_{\textrm{S}}$.
The disk mass $M_{ \text{D} }=1.8 \times M_{\textrm{c}}$.
The drag is strongest at the disk boundary, up to 3.46\%,
and weakest (circa 3.0\%) in the black region --- due to anti-dragging.
Here $\omega_0=1.31\sqrt{2GM_{\textrm{c}}}$.} 
\end{figure}

\section{Summary}

We investigated stationary gaseous disks around a spinless black hole in the 1PN approximation scheme. The 1PN calculation is obviously simpler technically than the full general-relativistic picture, but it does include essential features --- the nonlinearity and backreaction --- that are typical for Einstein equations.  The concepts of the quasilocal masses and angular momentum can be meaningfully defined in stationary axially symmetric systems, but it is useful that the 1PN approximation bases on Newtonian concepts and Newtonian intuition. 
The 1PN approximation allows for flexibility in imposing rotation laws that correspond to well known classes of Newtonian rotation curves, including the Keplerian rotation. We took care to deal with relatively small characteristic  velocities and potentials, $v/c\ll 1$ and  ${|U|}/c^2\ll1$, in all calculations reported here; that should guarantee that 1PN results are correct. The main results are following. 

An integrability condition leads to the emergence of two distinctly different types of general-relativistic corrections to the angular velocity.    

One of them is due to  the familiar geometric dragging of frames; this depends {\it only indirectly} on the disk structure and rotation,
because the metric in the 1PN expansion does include the backreaction and indirectly depends on characteristics of matter at the Newtonian level.  
The other  type  --- the dynamic effect --- that directly depends on the material structure of a disk, has been hitherto unknown.
That one is strongest in the central disk plane.  

The geometric and dynamic effects counteract; the dragging of frames pushes a disk forward, but the dynamic effect diminishes the angular velocity.
In  many  numerical solutions the net effect had been weakest in a central zone of the disk.
Inside of that zone the gas may be actually slowed down (the total 1PN angular velocity correction may or can be negative),
while at its boundary the dynamic  correction vanishes. The geometric dragging of frames always dominates in the disk boundary zone.  

The 1PN correction to the orbital period can be significant.
We found a number of numerical models that satisfy requirements of the 1PN approximation ---  with a  change of the orbital period close to 4\%. 
The Newtonian and 1PN equations posses a scaling symmetry,
that in principle would generate new solutions with even larger 1PN corrections to the angular velocity.
These rescaled solutions have to satisfy the assumptions of the 1PN approximations.
The ultimate  answer  concerning whether or not they approximate exact solutions
would require the investigation of an exact general-relativistic model.

We found that there is one type of matter --- dust --- for which the dynamic effect vanishes.
Disks made of dust experience 1PN corrections only through the dragging of frames.
It is interesting that uniformly rotating disks do not show any 1PN effect --- both contributory effects,
the geometric dragging of frames and the dynamic one, do vanish in this case.

In the merger of compact binaries consisting of pairs of black holes and neutron stars, a neutron star is  destroyed  \cite{Pan_Ton_Rez} and a heavy leftover  disk would form, that might reveal signs of the anti-dragging.   The  Bardeen-Petterson effect \cite{Bardeen-Petterson}, that arises  due to the geometric dragging,    is   known to occur in some AGN's \cite{Moran}. The dynamic braking  may lead    to its observable modifications    in black hole --- (heavy) disk systems.

\begin{acknowledgments}
This research was carried out with the supercomputer
``Deszno'' purchased thanks to the financial support of the European Regional
Development Fund in the framework of the Polish Innovation Economy
Operational Program (contract no.\ POIG.\ 02.01.00-12-023/08).
 The work of PJ was partially supported by the Polish NCN grant
\textit{Networking and R\&D for the Einstein Telescope}.  PM and MP acknowledge the support of the  Polish Ministry of Science and Higher Education grant IP2012~000172 (Iuventus Plus).
PM thanks Gerhard Sch\"afer for many discussions on the content of \cite{BDS} during his visit in Jena.
\end{acknowledgments}

\end{document}